\documentclass{article}

\usepackage{tikz}
\usepackage{srcltx} 
\usepackage{epsfig}
\usepackage[T1]{fontenc}
\usepackage[utf8]{inputenc}
\usepackage[english]{babel}
\usepackage{times}
\usepackage{color}
\usepackage{amsmath,enumerate}
\usepackage{amssymb}
\usepackage{amscd}
\usepackage{latexsym}
\usepackage{enumerate}
\usepackage{tabularx}
\usepackage{vmargin}
\usepackage{mathrsfs}
\usepackage{url}
\usepackage{hyperref}
\usepackage[affil-it]{authblk}
\usepackage{dsfont}

\usepackage{gensymb}

\newcommand{\mmdef}[1]{\emph{#1}}
\newcommand{\PP}{\mathcal{P}}
\newcommand{\II}[2]{[ #1 , #2 ]}

\DeclareMathOperator\MM{\mathcal{M}}
\DeclareMathOperator\RR{\mathbb{R}}
\DeclareMathOperator\DD{\mathsf{div}}

\DeclareMathOperator\TDD{\mathsf{TDD}}
\DeclareMathOperator\RD{\mathsf{RD}}

\title{Dushnik-Miller dimension of TD-Delaunay complexes\thanks{This
    research is partially supported by the ANR GATO, under contract
    ANR-16-CE40-0009.}}

\author[a]{Daniel Gonçalves}
\author[a]{Lucas Isenmann}

\affil[a]{{\small LIRMM, CNRS \& Université de Montpellier\\
     France\\
}}

\date{}

\newenvironment{proof}{\par \noindent \textbf{Proof.} }{\hfill$\Box$\medskip}

\newtheorem{theorem}{Theorem}

\newtheorem{definition}[theorem]{Definition}

\newtheorem{proposition}[theorem]{Proposition}
\newtheorem{lemma}[theorem]{Lemma}
\newtheorem{conjecture}[theorem]{Conjecture}

\newtheorem{example}[theorem]{Example}

\begin{document}

\maketitle

\begin{abstract}
TD-Delaunay graphs, where TD stands for triangular distance, is a variation of the classical 
Delaunay triangulations obtained from a specific convex distance function~\cite{CD}.
In \cite{BGHI} the authors noticed that every triangulation is the TD-Delaunay graph of a set of points in $\RR^2$, and conversely every TD-Delaunay graph is planar.  
It seems natural to study the generalization of this property in higher dimensions.  
Such a generalization is obtained by defining an analogue of the triangular distance for $\RR^d$.  
It is easy to see that TD-Delaunay complexes of $\RR^{d-1}$ are of Dushnik-Miller dimension $d$.  
The converse holds for $d=2$ or $3$ and it was conjectured to hold for larger $d$~\cite{Mary} (See also~\cite{EFKU}).  
Here we disprove the conjecture already for $d = 4$.
\end{abstract} 

\section{Introduction}

The order dimension (also known as the Dushnik-Miller dimension) of a poset $P$ has been introduced by Dushnik and Miller \cite{DM}.
It is the minimum number of linear extensions of $P$ such that $P$ is the intersection of these extensions.
See \cite{Trotter} for a comprehensive study of this topic.
Schnyder \cite{Schnyder} studied the Dushnik-Miller dimension of the incidence posets of graphs.
Some classes of graphs can be characterized by their Dushnik-Miller dimension which is the Dushnik-Miller dimension of their incidence poset.
For example, path forests are the graphs of Dushnik-Miller dimension at most $2$.
Schnyder~\cite{Schnyder} obtained a celebrated combinatorial characterization of planar graphs: they are those of Dushnik-Miller dimension at most $3$.
The question of characterizing classes of graphs of larger dimension is open.
Nevertheless there are some partial results.
Bayer {\it et al.}~\cite{BPS98} and Ossona de Mendez~\cite{Ossona} showed that every simplicial complex of Dushnik-Miller dimension $d$ has a straight line embedding in $\RR^{d+1}$ which generalizes the result of Schnyder in a way.
The reciprocal is false by considering $K_n$ which has a straight line embedding in $\RR^3$ while it has Dushnik-Miller dimension $\log\log n$ \cite{HM}.
The class of Dushnik-Miller dimension at most $4$ graphs is rather rich.
Extremal questions in this class of graphs have been studied: Felsner and Trotter~\cite{FT} showed that these graphs can have a quadratic number of edges.
Furthermore, in order to solve a question about conflict free coloring \cite{ELRS},  Chen \textit{et al.}~\cite{CPST} showed that most of the graphs of Dushnik-Miller dimension $4$ only have independent sets of size at most $o(n)$.
This result also implies that there is no constant $k$ such that every graph of Dushnik-Miller dimension at most $4$ is $k$-colorable.
Therefore, graphs of Dushnik-Miller dimension at most $4$ seem difficult to characterize. 
Nevertheless, it was conjectured in~\cite{Mary} (See also~\cite{EFKU}) that the class of Dushnik-Miller dimension $d$ complexes is the class of TD-Delaunay complexes which we will define in the next paragraph.
The result holds for $d=2$ and $d=3$, but we disprove it already for $d=4$ in this paper.

We now define the class of TD-Delaunay complexes which finds its origins in spanners introduced by Chew and Drysdale \cite{CD}.
Given points in the plane, a plane spanner is a subgraph of the complete graph on these points which is planar when joining adjacent points with segments.
The stretch of a plane spanner is the maximum ratio of the distance in the graph between two vertices when using the Euclidean weight function and the Euclidean distance between these two points.
Given points in the plane, the question raised by Chew and Drysdale is to find a plane spanner which minimizes the stretch.
We define the stretch of a class of plane spanners as the maximum stretch of any of these graphs.
Chew \cite{ChewL1} found the first class of plane spanners.
It consists in the class of $L_1$-Delaunay graphs which is a variant of the Delaunay graphs where the $L_2$ norm is replaced by the $L_1$ norm.
Given a norm $N$ and points $P$ in the plane, we define their Delaunay graph according to this norm as follows.
Given a point $x$, we define its Voronoi cell as the points $y$ of the plane such that $x$ is among the nearest points of $P$ (according to the norm $N$) to $y$.
Two points are connected if and only if their Voronoi cells intersect.
This is equivalent to the fact that there exists a $N$-disk containing both points but no other in its interior.
Chew~\cite{ChewL1} conjectured that the class of $L_2$-Delaunay graphs (classical Delaunay graphs) has a finite stretch.
This question initiated a series of papers about this topic which drops the upper bound from $5.08$ and $2.42$ to $1.998$~\cite{BK,DFS,Xia}.
A variant of this problem asks for minimizing the maximum degree.
Kanj \textit{et al.}~\cite{KPT} proved that there exists plane spanners of maximum degree 4 and of stretch at most $20$. 

Just a few years after introducing the plane spanner problem, Chew~\cite{ChewTD} found a second class of plane spanners.
It consists in TD-Delaunay graphs, obtained using the so-called triangular distance (which is not a distance but which is a convex distance function).
Given a compact convex shape $S$ and a point $c$ in the interior of $S$, we define the convex distance function, also called Minkowski distance function, between two points $p$ and $q$ as the minimal scaling factor $\lambda$ such that after rescaling $S$ by $\lambda$ and translating it in the way to center it on $p$, then it contains also $q$.
By taking $S$ the unit circle we get the Euclidean distance.
By taking $S$ the unit square we get the $L_{\infty}$ distance.
By taking $S$ an equilateral triangle we get what we call the triangular distance.
Chew~\cite{ChewTD} showed that their stretch is at most $2$ making the class of TD-Delaunay graphs the best plane spanners class until Xia~\cite{Xia} showed that the stretch of $L_2$-Delaunay graphs is strictly less than $2$.
This class is also used to obtain bounded degree plane spanners  $\cite{KPT}$.
TD-Delaunay graphs can be generalized to higher dimensions by taking the triangular distance in $\RR^d$ according to a regular $d$-simplex.

The second section is dedicated to the notion of Dushnik-Miller dimension applied to the inclusion poset of simplicial complexes.
In the third section we define TD-Delaunay complexes.
In the fourth section, we introduce the notion of multi-flows which will be useful to the main theorem.
The fifth section contains our main contribution in the form of a simplicial complex whose inclusion poset has Dushnik-Miller dimension $4$ but that is not a TD-Delaunay complex in $\RR^3$.
Finally, in the sixth section we prove that rectangular Delaunay complexes~\cite{CPST,FarXiv06} are TD-Delaunay complexes in $\RR^3$.

\section{Dushnik-Miller dimension of simplicial complexes}

Bonichon \textit{et al.} \cite{BGHI} showed the following property of TD-Delaunay graphs.
\begin{theorem}
    A graph $G$ is planar if and only if $G$ is a subgraph of a TD-Delaunay graph.
\end{theorem}

As a graph is planar if and only if $G$ is of Dushnik-Miller dimension at most $3$, there is maybe a link between the notions of Dushnik-Miller dimension and TD-Delaunay complexes.
In \cite{Mary} and \cite{EFKU} it was independently conjectured that Dushnik-Miller dimension at most $d$ complexes are exactly the subcomplexes of TD-Delaunay complexes of $\RR^{d-1}$.

We now need to define formally the notions used.
First of all abstract simplicial complexes generalize the notion of graphs.
An \emph{abstract simplicial complex} $\Delta$ with vertex set $V$ is a set of subsets of $V$ which is closed by inclusion (\textit{i.e.} $\forall Y \in \Delta$, $X \subseteq Y \Rightarrow X \in \Delta$). 
An element of $\Delta$ is called a \emph{face}.
A maximal element of $\Delta$ according to the inclusion order is called a \emph{facet}.

\begin{figure}[!h]
\centering
\includegraphics[scale=1]{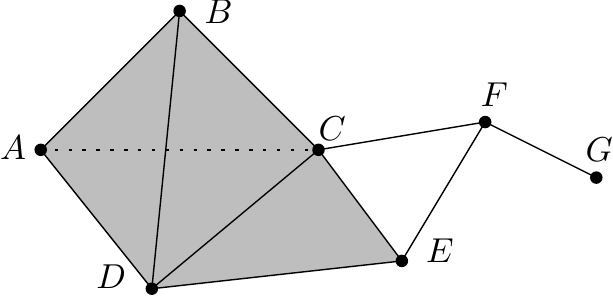}
\caption{An abstract simplicial complex whose facets are $\{A,B,C,D\},\{C,D,E\}, \{C,F\}, \{E,F\}, \{F,G\} $}
\label{fig:abs-simplicial-complex}
\end{figure}

\subsection{Dushnik-Miller dimension}

The notion of \emph{Dushnik-Miller} dimension of a poset has been introduced by Dushnik and Miller \cite{DM}.
It is also known as the \emph{order dimension} of a poset.
\begin{definition}
  The \emph{Dushnik-Miller dimension} of a poset $(\le,V)$ is the minimum number $d$ such that $(\le,V)$ is the intersection of $d$ linear extensions of $(\le,V)$. This means that there exists $d$ extensions $(\le_1,V), \ldots , (\le_d,V)$ of $(\le,V)$ such that for every $x,y \in V$, $x \le y$ if and only if $x \le_i y$ for every $i \in \II1d$.
     In particular if $x$ and $y$ are incomparable with respect to $\le$, then there exists $i$ and $j$ such that $x \le_i y$ and $y \le_j x$.
\end{definition}

\begin{figure}[!h]
	\centering
	\includegraphics[scale=1]{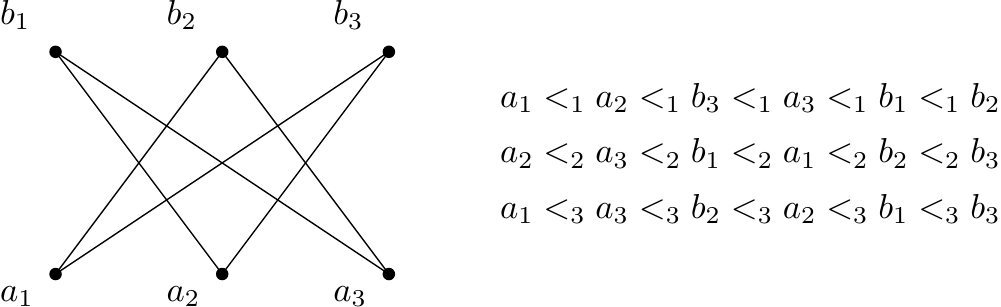}
	\caption{A poset of Dushnik-Miller dimension $3$ defined by its Hasse diagram and $3$ linear extensions $\le_1, \le_2$ and $\le_3$.}
\end{figure}

The notion of Dushnik-Miller dimension can be applied to an abstract simplicial complex as follows.

\begin{definition}
	Let $\Delta$ be an abstract simplicial complex.
    The inclusion poset of $\Delta$ is the poset $(\subset, \Delta)$.
	The Dushnik-Miller dimension of $\Delta$ is the Dushnik-Miller dimension of the inclusion poset of $\Delta$.
\end{definition}

\begin{figure}[!h]
\centering
\includegraphics{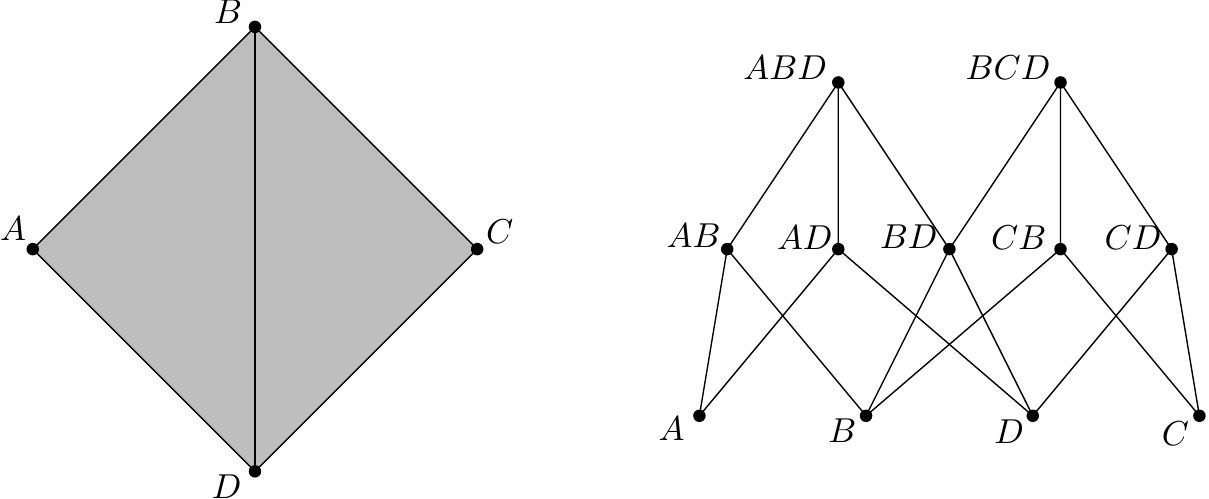}
\caption{An abstract simplicial complex with facets $\{A,B,D\}$ and $\{B,C,D\}$, and its inclusion poset.}
\label{fig:DM-dim}
\end{figure}

We denote $\dim_{DM}(\Delta)$ the Dushnik-Miller dimension of a simplicial complex $\Delta$.
Low dimensions are well known.
We have $\dim_{DM}(\Delta)=1$ if and only if $\Delta$ is a vertex.
We have $\dim_{DM}(\Delta)\le 2$ if and only if $\Delta$ is a union of paths.
There are complexes with arbitrarily high Dushnik-Miller dimension: for any integer $n$, $\dim_{DM}(K_n) = O( \log \log n)$ where $K_n$ denotes the complete graph on $n$ vertices.
The following theorem shows that the topological notion of planarity can be understood as a combinatorial property thanks to the Dushnik-Miller dimension.

\begin{theorem}[Schnyder~\cite{Schnyder}]
	\label{theorem:S89}
	A graph $G$ is planar if and only if $\dim_{DM}(G) \le 3$.
\end{theorem}

A generalization of the notion of planarity for simplicial complexes is the notion of straight line embedding.
As in the planar case, we do not want that two disjoint faces intersect.
We recall that a simplex of $\RR^d$ is the convex hull of a set of affinely independent points. We denote $\mathsf{Conv}(X)$, the convex hull of a set of points $X$.

\begin{definition}	
	Let $\Delta$ be a simplicial complex with vertex set $V$.
	A \emph{straight line embedding} of $\Delta$ in $\RR^d$ is a mapping $f : V \to \RR^d$ such that
	
	\begin{itemize}
	 \item $\forall X  \in \Delta$, $f(X)$ is a set of affinely independent points of $\RR^d$,
	 \item $\forall X,Y \in \Delta$, $\mathsf{Conv}(f(X)) \cap \mathsf{Conv}(f(Y)) = \mathsf{Conv}(f(X \cap Y))$.
	\end{itemize}
\end{definition}

The following theorem shows that the Dushnik-Miller dimension in higher dimensions also captures some geometrical properties.

\begin{theorem}[Bayer {\it et al.}~\cite{BPS98}, and Ossona de Mendez~\cite{Ossona}]
	\label{theorem:SLE}
	Any simplicial complex $\Delta$ such that\\ $\dim_{DM}(\Delta) \le d+1$ has a straight line embedding in $\RR^d$.
\end{theorem}

For $d=2$, this theorem states that if a simplicial complex has dimension at most $3$ then it is planar. 
Brightwell and Trotter \cite{BT} proved that the converse also holds (for $d=2$)\footnote{Note that in a straight line embedding in $\RR^2$ every triangle is finite, and it is thus impossible to embed a spherical complex like a octahedron or any polyhedron with triangular faces.}.
For higher $d$, the converse is false: take for example $K_n$ the complete graph which has a straight line embedding in $\RR^3$ (and therefore in $\RR^d$ for $d \ge 3$) and which has Dushnik-Miller dimension $O(\log \log n)$ \cite{HM}.

\subsection{Representations}

Representations have been introduced by Schnyder~\cite{Schnyder} in order to prove Theorem~\ref{theorem:S89}.
It is a tool for dealing with Dushnik-Miller dimension.
Here, only vertices will be ordered while in the Dushnik-Miller dimension every face of the complex must be ordered. 

\begin{definition}
	Given a linear order $\le$ on a set $V$, an element $x \in V$, and a set $F \subseteq V$, we say that $x$ {\it dominates} $F$ in $\le$, and we denote it $F\le x$, if $f \le x$ for every $f \in F$. 
	A {\it $d$-representation} $R$ on a set $V$ is a set of $d$ linear orders $\le_1 , \ldots , \le_d$ on $V$.	
	Given a $d$-representation $R$, an element $x \in V$, and a set $F \subseteq V$, we say that $x$ {\it dominates} $F$ (in $R$) if $x$ dominates $F$ in some order $\le_i\in R$.
	We define $\Sigma(R)$ as the set of subsets $F$ of $V$ such that every $v \in V$ dominates $F$.
\end{definition}

Note that $\emptyset\in \Sigma(R)$ for any representation $R$. An element $x\in V$ is a {\it vertex} of $\Sigma(R)$ if $\{x\}\subseteq V$. Note that sometimes an element $x\in V$ is not a vertex of $\Sigma(R)$. Actually, the definition of $d$-representation provided here is slightly different from the one in~\cite{Schnyder} and~\cite{Ossona}. There, the authors ask for the intersection of the $d$ orders to be an antichain. With this property, every element of $V$ is a vertex of $\Sigma(R)$. Note that simply removing the elements of $V$ that are not vertices of $\Sigma(R)$ yields a representation in the sense of~\cite{Ossona,Schnyder}.

\begin{proposition}
	For any $d$-representation $R = (\le_1, \ldots , \le_d)$ on a set $V$,
	$\Sigma(R)$ is an abstract simplicial complex.
\end{proposition}


\begin{proof}
	For any $F \in \Sigma(R)$, let $X$ be any subset of $F$, and let $v$ be any element of $V$.
	Since $F \in \Sigma(R)$, there exists $\le_i \in R$ such that $y \le_i v$ for every $y \in F$.	Particularly, $x \le_i v$ for every $x \in X$. Thus $X \in \Sigma(R)$ and we have proven that $\Sigma(R)$ is an abstract simplicial complex.
\end{proof}

An example is the following $3$-representation on $\{1,2,3,4,5\}$ where each line corresponds to a linear order whose elements appear in increasing order from left to right:

$$\begin{array}{c|cccccc}
   \le_1 & 1 & 2 & 5 & 4 & 3 \\
   \le_2 & 3 & 2 & 1 & 4 & 5 \\
   \le_3 & 5 & 4 & 3 & 2 & 1 
  \end{array}$$
  
The corresponding complex $\Sigma(R)$ is given by the facets $\{1,2\},
\{2,3,4\}$ and $\{2,4,5\}$.  For example $\{1,2,3\}$ is not in
$\Sigma(R)$ as $2$ does not dominate $\{1,2,3\}$ in any order.
 The following theorem shows that representations and Dushnik-Miller dimension are equivalent notions.
 
\begin{theorem}[Ossona de Mendez~\cite{Ossona}]
	\label{theorem:DM_representation}
	Let $\Delta$ be a simplicial complex with vertex set $V$.
	Then $\dim_{DM}(\Delta) \le d$ if and only if there exists a $d$-representation $R$ on $V$ such that $\Delta$ is included in $\Sigma(R)$.
\end{theorem}

For the following proofs, $R$ will be a $d$ representation $(\le_1, \ldots , \le_d)$ and $R'$ will be a $d$-representation $(\le_1', \ldots , \le_d')$.
The following lemmas are technical and will be useful for the proof of our main result, Theorem~\ref{thm:delaunay_ineq_syst}.

\begin{lemma}
	\label{lemma:repre_inversion}
	Let $R = (\le_1, \ldots , \le_d)$ be a $d$-representation on $V$.
	Let $x$ and $y$ be two different vertices of $\Sigma(R)$ (i.e. $\{x\}$ and $\{y\} \in \Sigma(R)$) such that $\{x,y\} \not\in \Sigma(R)$ and such that $x$ and $y$ are consecutive in the order $\le_i$.
	The representation obtained after the permutation of $x$ and $y$ in the order $i$, denoted $R'=(\le'_1, \ldots, \le'_d)$, is such that $\Sigma(R') = \Sigma(R)$. 
\end{lemma}
\begin{proof}
	Without loss of generality, we suppose that $i = 1$ and that $x \le_1 y $.
	Let us first show that $\Sigma(R') \subseteq \Sigma(R)$.
	Towards a contradiction, let us consider a face $F \in \Sigma(R)$ such that $F \not\in \Sigma(R')$.
	There exists therefore a vertex $z \in V$ which does not dominate $F$ in $R'$.
	As $z$ dominates $F$ in $R$, and as $\le_i = \le'_i$ for every $i\neq 1$, we thus have that $F \le_1 z$ and $F \not\le_1' z$.	This implies that $z = y$, that $x \in F$. As $x$ and $y$ are not adjacent, we have that $y\notin F$.
	Furthermore, $y$ only dominates $F$ in order $\le_1$ of $R$.
	We denote $f_i$ the maximum vertex of $F$ in order $\le_i$, then $f_1 = x$ and:

	$$\begin{array}{c|ccccccc}
	\le_1 & \cdots & x & y & \cdots \\
	\le_2 & \cdots & y & \cdots & f_2 & \cdots \\
	\vdots & & \vdots & & \vdots \\
	\le_d & \cdots & y & \cdots & f_d & \cdots \\
	\end{array}$$
	
	As $\{x,y\} \not\in \Sigma(R)$, there exists an element $w\in V$ such that $w <_i \max_{\le_i}(x,y)$ for every $i$.
	Thus either $w=x$, contradicting the fact that $\{y\}\in \Sigma(R)$, or $w <_i f_i$ for every $i$ (in particular for $i=1$ because $x$ and $y$ are consecutive), contradicting the fact that $F \in \Sigma(R)$.
		
	We showed that $\Sigma(R) \subseteq \Sigma(R')$,
	but as this operation is an involution, we have that $\Sigma(R') \subseteq \Sigma(R)$.
	Thus $\Sigma(R) = \Sigma(R')$.
\end{proof}

Note that $R'$ is such that for every edge $ab\in \Sigma(R)$ the orders between its endpoints are preserved. In other words, for any $j$ we have that $a\le'_j b$ if and only if $a\le_j b$.
Given a $d$-representation $R$, an {\it $(\le_i)$-increasing $xy$-path} is a path $(a_0=x, a_1, a_2 , \ldots , a_k=y)$ in $\Sigma(R)$ such that $a_j \le_i a_{j+1}$ for every $0\le j<k$. 

\begin{lemma}
    \label{lemma:increasing_path}
	Let $R$ be a $d$-representation on $V$, let $x$ be a vertex of $\Sigma(R)$ (i.e. $\{x\} \in \Sigma(R)$), and let $\le_i$ be any order of $R$. There exists a $d$-representation $R'$ such that:
	\begin{itemize}
	\item $\Sigma(R')=\Sigma(R)$,
	\item for every $j$, a path $P$ is $(\le'_j)$-increasing (in $\Sigma(R')$) if and only if it is $(\le_j)$-increasing (in $\Sigma(R)$), 
	\item $\le'_j \;\equiv\; \le_j$ for every $j\neq i$,
	\item $x \le'_i y$ if and only if there exists an $(\le'_i)$-increasing $xy$-path in $\Sigma(R')$ (thus if and only if there exists an $(\le_i)$-increasing $xy$-path in $\Sigma(R)$), and
	\item $a\le'_i b\le'_i x$ implies that $a\le_i b$.
	\end{itemize} 
\end{lemma}
\begin{proof}
    First note that it is sufficient to prove the second item for length one paths (as longer paths are just concatenations of length one paths).
    We proceed by induction on $n$, the number of couples $(y,z)$ such that $x\le_i y <_i z$, such that there exists an $(\le_i)$-increasing $xy$-path, and such that there is no $(\le_i)$-increasing $xz$-path. In the initial case, $n=0$, as for every vertex $z$ such that $x\le_i z$ there exists an $(\le_i)$-increasing $xz$-path, we are done and $R'=R$.

    If $n>0$, consider such couple $(y,z)$ with the property that $y$ and $z$ are consecutive in $\le_i$ (by taking $z$ as the lowest element in $\le_i$ such that there is no $(\le_i)$-increasing $xz$-path).
    Note that $\{y,z\}\notin \Sigma(R)$ as otherwise, extending an $(\le_i)$-increasing $xy$-path with the edge $yz$ one obtains an $(\le_i)$-increasing $xz$-path.
	By Lemma~\ref{lemma:repre_inversion}, the $d$-representation $R'$ obtained by permuting $y$ and $z$ is such that $\Sigma(R') = \Sigma(R)$, such that the orders between the endpoints of any edge are preserved (and thus the increasing paths are preserved), such that $\le'_j \;\equiv\; \le_j$ for every $j\neq i$, such that for any two vertices $a$ and $b$ without $(\le'_i)$-increasing $xa$-path nor $xb$-path, $a\le'_i b$ if and only if $a\le_i b$, and has only $n-1$ couples $(y',z')$.
	
    We can thus apply the induction hypothesis to $R'$ and we obtain that there exists a $d$-representation $R''$ such that $\Sigma(R'')=\Sigma(R')=\Sigma(R)$, such that the increasing paths are the same as in $\Sigma(R')$ (and thus as in $\Sigma(R)$), such that $\le''_j \;\equiv\; \le'_j \;\equiv\; \le_j$ for every $j\neq i$, such that  $x \le''_i y$ if and only if there exists an $(\le''_i)$-increasing $xy$-path in $\Sigma(R'')$, and such that $a\le''_i b\le''_i x$ implies that $a\le'_i b$, which implies $a\le_i b$ (as there is no $(\le'_i)$-increasing $xa$-path nor $xb$-path).
\end{proof}

\begin{lemma}
    \label{lemma:repre_path}
    Let $R$ be a $d$-representation on $V$.
    For any face $F \in \Sigma(R)$ and any vertex $x$ of $\Sigma(R)$, there exists an  $(\le_i)$-increasing $f_ix$-path for some order $\le_i\in R$, where $f_i$ is the maximal vertex of $F$ in order $\le_i$.
\end{lemma}
\begin{proof}
    We proceed by induction on $n$, the number of orders $\le_i \in R$ such that $  F \le_i x$.
	If $n=0$ then $F\notin \Sigma(R)$, a contradiction. So the lemma holds by vacuity.
	If $n>0$, consider an order $\le_i \in R$ such that $F \le_i x$. By Lemma~\ref{lemma:increasing_path} (applied to $f_i$ in $\le_i$) either $\Sigma(R)$ contains an $(\le_i)$-increasing $f_ix$-path, and we are done, or there exists a $d$-representation $R'$ such that $\Sigma(R')=\Sigma(R)$, such that the increasing paths are the same as in $\Sigma(R)$, such that $\le'_j \;\equiv\; \le_j$ for every $j\neq i$, and such that $x <'_i f_i$. In this case, we apply the induction hypothesis on $R'$. Indeed, $F \in \Sigma(R')$ and $R'$ has only $n-1$ orders $\le'_j$ such that $F \le'_j x$. Note that as every pair in $F$ corresponds to an edge, the maximal vertices of $F$ in $\le'_j$ and $\le_j$ are the same for every $j$. The induction thus provides us an $(\le'_j)$-increasing $f_jx$-path and this path is also $(\le_j)$-increasing in $\Sigma(R)$.
\end{proof}

\begin{lemma}
    \label{lemma:path_to_nonface}
    Let $R$ be a $d$-representation on $V$.
    For any vertex set $F \not\in \Sigma(R)$, there exists a vertex $x$  which does not dominate $F$, and 
    such that for every $\le_i \in R$ there exists an $(\le_i)$-increasing $xf^i$-path for some vertex $f^i\in F$.
\end{lemma}
\begin{proof}
	We proceed by induction on the number $n$ of elements which do not dominate $F$.
	If $n=0$ then $F\in \Sigma(R)$, a contradiction. So the lemma holds by vacuity.
    If $n>0$, consider any such element $x$ which does not dominate $F$.
    By Lemma~\ref{lemma:increasing_path} either there are $(\le_i)$-increasing $xf^i$-paths for every $\le_i\in R$ and we are done, or there exists a $d$-representation $R'$ such that $\Sigma(R')=\Sigma(R)$, 
    such that the increasing paths are the same as in $\Sigma(R)$,
    such that $F \le'_i x$ for some order $\le'_i \in R'$, 
    such that $a\le'_i b\le'_i x$ implies that $a\le_i b$,
    and such that $\le'_j \;\equiv\; \le_j$ for every $j\neq i$.
    
    In this case, we apply the induction hypothesis on $R'$. Indeed, as any element not dominating $F$ in $R'$ does not dominate $F$ in $R$, and as $x$ is not dominating $F$ in $R'$, we have at most $n-1$ elements that do not dominate $F$ in $R'$. By induction hypothesis there exists a vertex $x$ such that for every $\le'_i \in R'$, $\Sigma(R')$ has an $(\le'_i)$-increasing $xf^i$-path for some vertex $f^i\in F$, and this path is also $(\le'_i)$-increasing in $\Sigma(R)$. 
\end{proof}

\section{TD-Delaunay complexes}

TD-Delaunay graphs have been introduced by Chew in \cite{ChewTD}.
Here we generalize this definition to higher dimensions.
We recall that a positive homothety $h$ of $\RR^d$ is an affine transformation of $\RR^d$ defined by $ h(M) = \alpha M + (1-\alpha) \Omega $ where $\alpha \ge 0$ and $\Omega \in \RR^d$.
The coordinates of a point $x \in \RR^d$ will be denoted $(x_1, \ldots , x_d)$.
	
For any integer $d$, let $H_d$ be the $(d-1)$-dimensional hyperplane of $\RR^d$ defined by $\{ x \in \RR^d : x_1 + \cdots + x_d = 1 \}$. 
Given $c = (c_1, \ldots , c_d) \in \RR^d$, we define a \mmdef{regular simplex} $S_c$ of $H_d$ by setting $S_c = \{ u \in H_d : u_i \le c_i, \forall i \in \II1d \}$.
A regular simplex $S_c$ is said to be \mmdef{positive} if $\sum_{i=1}^d c_i \ge 1$.
For $c = (1,\ldots , 1) = \mathds{1}$ we call $S_\mathds{1}$ the \emph{canonical regular simplex}.
In this context, a point set $\mathcal{P} \subset H_d$ is in
\mmdef{general position} if for any two vertices $x,y\in \mathcal{P}$,
$x_i \not= y_i$ for every $i \in \II1d$.
The \mmdef{interior} $\mathring{S_c}$ of a regular simplex $S_c$ is defined by $\mathring{S_c} = \{ u \in H_d : u_i < c_i , \forall i \in \II1d \}$

\begin{proposition}
	The positive regular simplices of $H_d$ are the subsets of $H_d$ positively homothetic to $S_\mathds{1}$.
\end{proposition}
\begin{proof}
    Let $U$ be a subset of $H_d$ positively homothetic to $S_\mathds{1}$.
    Let us show that there exists $c \in \RR^d$ such that $U = S_c$ and $\sum_{i=1}^d c_i \ge 1$.
    There exists $h$ a positive homothety of $\RR^d$ of ratio $\alpha$ and center $\Omega \in H_d$ such that $h(U) = S_\mathds{1}$.
	We define $c = (c_1, \ldots , c_d) \in \RR^d$ by $c_i = \alpha + (1 - \alpha)\Omega_i$.
	We show that $h(S_\mathds{1}) = S_c$.
	If $\alpha = 0$ it is clear that $h(S_\mathds{1}) = \{\Omega\}$ and that $S_c = \{\Omega\}$.
	We can therefore suppose that $\alpha >0$.
	
	Let $u \in S_\mathds{1}$.
	Then $u_i \le 1$ for every $i \in \II1d$.
	Then $h(u)_i = \alpha u_i + (1-\alpha) \Omega_i \le \alpha + (1-\alpha) \Omega_i = c_i$.
	Furthermore $\sum_{i=1}^d h(u)_i = \alpha \sum_{i=1}^d u_i + (1- \alpha) \sum_{i=1}^d \Omega_i = \alpha + (1- \alpha) = 1$.
	Then $h(u) \in S_c$ and thus $h(S_\mathds{1}) \subseteq S_c$.
	
	Let $v \in S_c$.
	Then $v_i \le c_i$ for every $i \in \II1d$.
	Because $\alpha \not=0$, we define $u$ such that $h(u) = v$, that is such that $\alpha u_i + (1- \alpha) \Omega_i = v_i$. 
	As $ 1 = \sum_{i=1}^d v_i = \alpha \sum_{i=1}^d u_i + (1-\alpha) \sum_{i=1}^d \Omega_i = \alpha \sum_{i=1}^d u_i + (1- \alpha)$, then $u \in H_d$.
	Furthermore for every $i \in \II1d$, $ \alpha u_i + (1-\alpha) \Omega_i = v_i \le \alpha + (1-\alpha) \Omega_i$.
	Thus $u_i \le 1$.
	So, $u \in S_\mathds{1}$.
	
	We conclude that $h(S_\mathds{1}) = S_c$.
	Furthermore if $\alpha \ge 0$ then $S_c$ is positive.
	Indeed $\sum_{i=1}^d c_i = d \alpha + (1- \alpha) \sum_{i=1}^d \Omega_i = (d-1) \alpha + 1 \ge 1$.
	We deduce that every subset of $H_d$ which is positively homothetic to $S_\mathds{1}$ is a positive regular simplex.
	
	Let $S_c$ be a positive regular simplex of $H_d$ with $c \in \RR^d$ such that $\sum_{i=1}^d c_i \ge 1$.
	We look for $\alpha \ge 0$ and $\Omega \in H_d$ such that $c = \alpha + (1-\alpha) \Omega$.
	Suppose that such an $\alpha \ge 0$ and an $\Omega$ exist.
	Then $c_i = \alpha + (1-\alpha) \Omega_i$ for every $i \in \II1d$.
	Then $\sum_{i=1}^d c_i = (d-1) \alpha + 1$.
	Thus $\alpha = (\sum_{i=1}^d c_i -1)/(d-1) \ge 0$ and $\Omega_i = (c_i - \alpha)/(1-\alpha)$.
	It is easy to check that this gives the desired $\alpha$ and $\Omega$ and that they are well defined even if $\alpha = 0$ or $1$.
	We conclude that $S_c$ is positevly homothetic to $S_\mathds{1}$ in $H_d$.
\end{proof}


	
	
Let us now define TD-Delaunay simplicial complexes by extending the notion of TD-Delaunay graph defined by Chew and Drysdale~\cite{CD}.

\begin{definition}
	Given a set $\PP$ of points of $H_d$ ($\subset \RR^d$) in general position, let the TD-Delaunay complex of $\mathcal{P}$, denoted $\TDD( \mathcal{P})$, be the simplicial complex with vertex set $\mathcal{P}$ defined as follows.
	A subset $F \subseteq \PP$ is a face of $\TDD(\PP)$ if and only if there exists a positive regular simplex $S$ such that $S \cap \PP = F$ and such that no point of $\PP$  is in the interior of $S$.
\end{definition}

Let $F \subseteq \PP$, we define $c^F \in \RR^d$ by $c^F_i = \max_{x \in F } x_i$.
Remark that $\sum_{i=1}^d c^F_i \ge 1$ because for every $x \in F$, $\sum_{i=1}^d c^F_i \ge \sum_{i=1}^d x_i = 1$.

\begin{lemma}
    \label{lemma:face_caracterization}
    For any set $F \subseteq \PP$, $F \in \TDD(\PP)$ if and only if $S_{c^F}$ does not contain any point of $\PP$ in its interior.
\end{lemma}
\begin{proof}
    Suppose that $F \in \TDD(\PP)$.
    Then there exists a positive regular simplex $S_c$ such that $S \bigcap \PP = F$ and such that no point of $\PP$ is in the interior of $S$.
    As $S_c$ contains $F$, for every $x \in F$ and every $i \in \II1d$, $x_i \le c_i$.
    Thus $c^F_i \le c_i$ for every $i$ and $S_{c^F} \subseteq S_c$.
    Therefore $S_{c^F}$ does not contain any point of $\PP$ in its interior otherwise $S_c$ would contain some.

    Suppose that $S_{c^F}$ does not contain any point of $\PP$ in its interior.
    Let $x \in F$.
    By definition of $c^F$, $x_i \le c^F_i$ for every $i$.
    So $x \in S_{c^F}$ and $S_{c^F}$ contains $F$.
    Let $x \in \PP \setminus F$.
    If $x$ is in $S_{c^F}$, then $x$ is not in the interior of $S_{c^F}$.
    So there exists $i \in \II1d$ such that $x_i = c^F_i$.
    But there exists $y \in F$ (different from $x$) such that $y_i = c^F_i$.
    This contradicts the fact that the points of $\PP$ are in general position.
    Therefore $S_{c^F} \cap \PP = F$ and $F \in \TDD(\PP)$.
\end{proof}

\begin{proposition}
	For any point set $\PP$ in general position in $H_d$ ($\subset \RR^d$), $\TDD(\PP)$ is an abstract simplicial complex.
\end{proposition}
\begin{proof}
    Consider any $F \in \TDD(\PP)$, and any $G \subsetneq F$.
    Then $c^G_i \le c^F_i$ for every $i$.
   So $S_{c^G} \subseteq S_{c^F}$.
   $S_{c^G}$ does not contain any point of $\PP$ in its interior otherwise $S_{c^F}$ would contain some.
   Thus because of the previous lemma, $G \in \TDD(\PP)$ and we conclude that $\TDD(\PP)$ is an abstract simplicial complex.
\end{proof}

\begin{figure}[!h]
	\centering
	\includegraphics[scale=0.8]{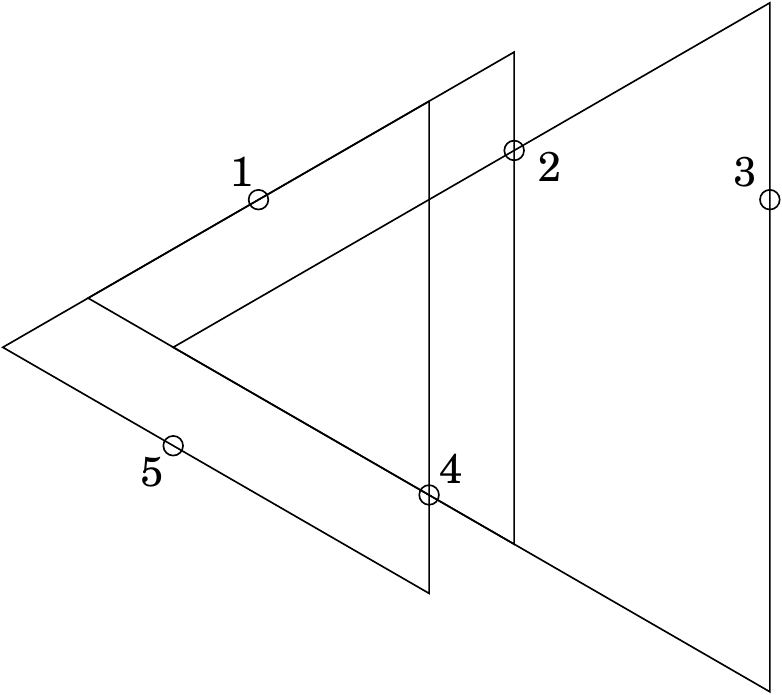}
	\caption{An example of a point set which TD-Delaunay complex is the simplicial complex with facets $\{1,4,5\}, \{1,2,4\}$ and $\{2,3,4\}$.}
\end{figure}

Consider a point set $\PP$ of $\RR^d$ in general position.
We define the orders $\le_i$ on $\PP$ as $x \le_i y$ if and only if $x_i \le y_i$, in other words if and only if $x_i - y_i \le 0$. First, note that as the points are in general position these orders are well defined. Note also that the values $x_i$ for $x\in \PP$ and $i\in \II1d$ form a solution of a linear system of inequalities (with inequalities of the form $x_i - y_i \le 0$).
In the following we connect TD-Delaunay complexes to representations through systems of inequalities.
To do so we consider these $d$ orders as a $d$-representation denoted $R(\PP)$. If $\PP\subset H_d$ this $d$-representation is closely related to $\TDD(\PP)$.

\begin{theorem}\label{thm:TDDisDMd}
For any point set $\PP$ of $H_d\subset \RR^{d}$ in general position, we have that $\TDD(\PP) = \Sigma(R(\PP))$. Thus, any TD-Delaunay complex of $H_d\simeq \RR^{d-1}$ has Dushnik-Miller dimension at most $d$.
\end{theorem}
\begin{proof}
Consider a set of points $\PP$ of $H_d$ in general position and let $R(\PP) = (\le_1,\ldots,\le_d)$ be the $d$-representation on $\PP$ such that $u\le_i v$ if and only if $u_i \le v_i$. 

Let us first prove that $\TDD(\PP) \subseteq \Sigma(R(\PP))$, by showing that for any $F\in \TDD(\PP)$ we have that $F\in \Sigma(R(\PP))$. By definition there exists $c \in \RR^{d}$ such that $S_c$ contains exactly the points $F$, and they lie on its border.  
For every $i$, we denote by $f_i$ the maximum among the elements of $F$ with respect to $\le_i$.  
Towards a contradiction we suppose that $F \not\in \Sigma(R(\PP))$. 
Thus there exists a vertex $z$ of $\PP$ such that $z$ does not dominate $F$ in any order of $R(\PP)$. Thus $z <_i f_i$ for every $i$. 
Therefore $z_i < (f_i)_i$ for every $i$. 
But $(f_i)_i \le c_i$ because $f_i \in S_c$ thus $z_i < c_i$.
Hence $z \in \mathring{S_c}$ contradicting the definition of $S_c$.  

Let us now prove that $\Sigma(R(\PP)) \subseteq \TDD(\PP)$, by showing that for any $F\in \Sigma(R(\PP))$ we have that $F\in \TDD(\PP)$. Consider any non-empty face $F \in \Sigma(R)$ (the case of the empty face is trivial) and suppose towards a contradiction that $F\not\in \TDD(\PP)$. According to Lemma~\ref{lemma:face_caracterization}, there exists $x \in \PP$ such that $x \in \mathring{S}_{c^F}$. Thus $x_i < c^F_i$ for every $i$. For every $i$, we define $f_i$ as the maximum among the elements of $F$ with respect to $\le_i$. Thus $x_i < c^F_i = (f_i)_i$, and $x <_i f_i$ for every $i$, which contradicts the fact that $F \in \Sigma(R)$.
\end{proof}

The reciprocal statement holds for $d=2$
and 3, as any $d$-representation $R$ is such that $\Sigma(R)$ defines a TD-Delaunay complex of $H_d$. This naturally raised the following conjecture~\cite{Mary} (See also in~\cite{EFKU} as an open problem). 

\begin{conjecture}[\cite{Mary}]\label{conj}
For every $d$-representation $R$, the abstract simplicial complex
$\Sigma(R)$ is a TD-Delaunay complex of $H_d\simeq \RR^{d-1}$.
\end{conjecture}

Mary~\cite{Mary} proved the conjecture when $\Sigma(R)$
already admits some particular embedding. In the following we show
that actually this conjecture does not hold, already for $d=4$.
To do so, in the following we characterize which representations $R$ are such that $\Sigma(R)$ is a TD-Delaunay complex. Actually we are first going to characterize which $d$-representations $R$ correspond to some point set $\PP$ of $H_d$ such that $R=R(\PP)$, and then we are going to show that for those we have $\TDD(\PP) = \Sigma(R(\PP))$.
By definition of $R(\PP)$, for any different $x,y \in \PP$ and for any order $\le_i\in R(\PP)$ we have that $x\le_i y$, if and only if 
\begin{equation}~\label{eq:xyi}
y_i - x_i > 0
\end{equation}
Furthermore as we consider a point set $\PP$ of $H_d$ we have that $\sum_{1\le i\le d} x_i =1$, which gives:
\begin{equation}~\label{eq:Hd}
x_d = 1-\sum_{1\le i< d} x_i
\end{equation}
In the following we consider a $d$-representation $R$, and we define the system of inequalities obtained by taking Inequality~(\ref{eq:xyi}) for every $i\in \II1d$ but only for the pairs $\{x,y\}\in \Sigma(R)$, and by replacing the $d^{th}$ coordinates by the right hand of Equation~(\ref{eq:Hd}).

\begin{definition}[TD-Delaunay system]
    	Let $R$ be a $d$-representation on a vertex set $V$, and consider the edge set $E$ of $\Sigma(R)$ defined by $E = \{ X \in \Sigma(R) : |X| = 2\}$.
    	We define the matrix $A_R$ of $\MM_{ E \times \II1d , V \times \II1{d-1}}(\RR)$
    	where the coefficients, $a_{(e,i),(v,j)}$ of $A_R$ are indexed by an edge $e\in E$, a vertex $v\in V$, and two indices $i \in \II1d$ and $j \in \II1{d-1}$.
    \[
    a_{(e=\{x,y\},i),(v,j)} = 
    \begin{cases}
    +1 &\text{ if $i=j$, $v \in e$ and $v = \max_{\le_j} (x,y)$} \\
    -1 & \text{ if $i=j$, $v \in e$ and  $v = \min_{\le_j} (x,y)$}  \\
    +1 &\text{ if $i = d$, $v \in e$ and  $v = \min_{\le_d} (x,y)$}  \\
    -1 &\text{ if $i = d$, $v \in e$ and  $v = \max_{\le_d} (x,y)$}  \\
    0 &\text{ otherwise}
    \end{cases}
    \]
    The \mmdef{TD-Delaunay system} of the representation $R$ is the following linear system of inequalities :
	\[
	 A_R X > 0 
	 \]
	for some vector $X \in \RR^{V \times \II1{d-1}}$.
\end{definition}

\begin{example}
	We consider the following $3$-representation $R$ on $V=\{a,b,c\}$:
	
	$$ \begin{array}{c|cccc} \le_1 & b & c & a \\ \le_2 & a & c & b \\ \le_3 & a & b & c
	   \end{array}$$
	The complex $\Sigma(R)$ is given by the facet $\{a,b,c\}$ and contains $3$ edges: $ab, bc$ and $ac$.
	The matrix of the TD-Delaunay system of $R$ is:
    \[
	A_R = \bordermatrix{
	& (a,1) & (b,1) & (c,1) & (a,2) & (b,2) & (c,2) \cr
	(bc,1) & & -1 & 1 \cr
	(ac,1) & 1 & & -1 \cr
	(ab,1) & 1 & -1 & \cr
	(bc,2) & &&& & 1 & -1 \cr
	(ac,2) & &&& -1 & & 1 \cr
	(ab,2) & &&& -1 & 1 & \cr
	(bc,3) & & 1 & -1 & & 1 & -1 \cr
	(ac,3) & 1 & & -1& 1 & & -1 \cr
	(ab,3) & 1 & -1  & & 1 & -1 & \cr
	}
	\]
	The system $A_R X > 0$ where $X \in \RR^{V \times \II1{d-1}}$ is equivalent to the following linear system, where $v_i$ denotes $X_{(v,i)}$. 
	\[
	\begin{cases}
		b_1 < c_1 \\
		c_1 < a_1 \\
		b_1 < a_1 \\
		c_2 < b_2 \\
		a_2 < c_2 \\
		a_2 < b_2 \\
		c_1 + c_2 < b_1 + b_2 \\
		c_1 + c_2 < a_1 + a_2 \\
		b_1 + b_2 < a_1 + a_2 
	\end{cases}
	\]
	Note that setting $a_3$, $b_3$, and $c_3$ to $1-a_1-a_2$, $1-b_1-b_2$, and $1-c_1-c_2$ respectively, the last three equations imply that $b_3< c_3$, $a_3< c_3$ and $a_3< b_3$.
\end{example}

\begin{theorem}\label{thm:delaunay_ineq_syst}
For any abstract simplicial complex $\Delta$ with vertex set $V$, $\Delta$ is a TD-Delaunay complex of $H_d \simeq \RR^{d-1}$ if and only if there
exists a $d$-representation $R$ on $V$ such that $\Delta = \Sigma(R)$
and such that the corresponding TD-Delaunay system has a solution.
\end{theorem}

\begin{proof}($\Rightarrow$)
This follows from Theorem~\ref{thm:TDDisDMd} and from the fact that the coordinates of any point set $\PP$ form a solution to the TD-Delaunay system defined by $R(\PP)$.
	
($\Leftarrow$) Consider now a $d$-representation $R$ on a set $V$ such that the TD-Delaunay system of $R$ has a solution $X \in \RR^{V \times \II1{d-1}}$ and for any $v\in V$ let us define a point $v\in\RR^d$ by setting $v_i = X_{(v,i)}$ for $i\in \II1{d-1}$, and $v_d = 1-v_1-\ldots -v_{d-1}$. 
This implies that all these points belong to $H_d$. 
It will be clear from the context when we refer to an element of $V$ or to the corresponding point of $\RR^d$.
As in the linear system the inequalities are strict, one can slightly perturb the position of the vertices in order to obtain points in general position that still fulfill the system constraints.
Recall that by construction, for any edge $uv$ of $\Sigma(R)$ and any $i \in \II1d$, $u_i < v_i$ if and only if $u <_i v$. This implies that if $\Sigma(R)$ has an $(\le_i)$-increasing $xy$-path then $x_i \le y_i$.

Let us first prove that $\TDD(\PP) \subseteq \Sigma(R)$.
Consider a face $F \in \TDD(\PP)$ and suppose towards a contradiction that $F \not\in \Sigma(R)$. 
As $F \in \TDD(\PP)$, there exists $c\in \RR^{d}$ such that $S_c$ contains exactly the points $F$, and they lie on its border. As $F \not\in \Sigma(R)$, then by Lemma~\ref{lemma:path_to_nonface}, there exists $x\in V$ which does not dominate $F$, and such that for every $\le_i \in R$ there exists an $(\le_i)$-increasing $xf^i$-path for some vertex $f^i\in F$. Therefore $x_i \le (f^i)_i \le c_i$. As $x$ does not dominate $F$ and as the points are in general position one of these inequalities is strict and we conclude that $x$ lies in the interior of $S_c$, a contradiction.

Let us now prove that $\Sigma(R) \subseteq \TDD(\PP)$. Consider any non-empty face $F \in \Sigma(R)$. For every $i \in \II1d$, we denote $f_i$ the maximum of $F$ in the order $\le_i$ and we define $c\in \RR^d$ (and $S_c$) by setting $c_i = (f_i)_i$.  
First note that for any vertex $u \in F$ and any $i \in \II1d$, as $uf_i$ is an edge, we have that $c_i = (f_i)_i \ge u_i$. We hence have that $\sum_{i=1}^{d} c_i \ge \sum_{i=1}^{d} u_i = 1$. 
As $F \in \Sigma(R)$, for every $u \in F$ there exists an $i \in \II1d$ such that $u = f_i$.
Therefore $u_i = (f_i)_i = c_i$ and as $u_j \le c_j$ for every $j$ (because either $u=f_j$ or $\{u,f_j\} \in \Sigma(R)$ and then $u_j \le (f_j)_j = c_j$), we have that $u$ is on the border of $S$.  
According to Lemma~\ref{lemma:repre_path}, for every $u \not\in F$, there exists an $(\le_i)$-increasing $f_iu$-path in $\Sigma(R)$, for some order $\le_i \in R$. Therefore $c_i = (f_i)_i < u_i$ and $u \not\in S_c$.  
Thus $F \in \TDD(\PP)$ and $\Sigma(R) \subseteq \TDD(\PP)$.
%
\end{proof}

\section{Multi-flows}

We disprove Conjecture~\ref{conj} using  Theorem~\ref{thm:delaunay_ineq_syst} by exhibiting a simplicial complex $\Delta$, such that  $\Delta = \Sigma(R)$ for some 4-representation $R$, and such that none of the 4-representation $R'$ verifying $\Delta = \Sigma(R')$ admits a solution to its TD-Delaunay system. A common tool to prove that a system of inequalities has no solution is the celebrated Farkas lemma.

\begin{lemma}[Farkas lemma]\label{lemma:farkas}
    For any $ m \times n $ real matrix $A$, either
    \begin{itemize}
     \item $Ax > 0$ admits a solution $x\in \RR^m$, or
     \item ${}^t A y = 0$ admits a non-zero solution $y\in (\RR^+)^n$.
    \end{itemize}
    Furthermore both cases are exclusive.
\end{lemma}

In the following we show that this lemma defines a dual notion of a TD-Delaunay solution, we call it a {\it multi-flow}. To define it, we first need to recall some notions about flows.
Let $G=(V,A)$ be a digraph with vertex set $V$ and arc set $A \subseteq V \times V$.
A flow on $G$ is a function of $\varphi : A \to \RR^+$.  
Let the divergence $\mathsf{div}_\varphi(x)$ of a vertex $x$ be given by $\mathsf{div}_\varphi(x) = \sum_{(y,x) \in A} \varphi(y,x) - \sum_{(x,y) \in A} \varphi(x,y)$.

\begin{definition}[Multi-flow]
	Let $R=\{\le_1,\ldots,\le_d\}$ be a $d$-representation on $V$.
	For $i \in \II1d$, $G^i(R)$ will denote the digraph with vertex set $V$ and arc set $A^i$ where $A^i = \{ (x,y) \in V\times V : \{x,y\} \in \Sigma(R) , x  \le_i y \}$.
	A \mmdef{multi-flow} is a collection of $d$ flows $\varphi_1, \ldots, \varphi_d$ respectively on each digraph $G^1(R), \ldots , G^d(R)$ such that $\DD_{\varphi_i}(v) = \DD_{\varphi_d}(v)$, for every $v\in V$ and every $i \in \II1d$.
\end{definition}

We can now state the main result of this section.

\begin{proposition}
	\label{proposition:alternative}
	Let $R$ be a $d$-representation on a set $V$.
	Then either
	\begin{itemize}
	 \item the corresponding TD-Delaunay  system admits a solution, or
	 \item $R$ admits a non-zero multi-flow (i.e. a multi-flow with
           some $\varphi_i(uv)>0$).
	\end{itemize}

	Furthermore both cases are exclusive.
\end{proposition}
\begin{proof}
	According to Proposition~\ref{thm:delaunay_ineq_syst}, the corresponding TD-Delaunay system admits a solution if and only if $A_R x > 0$ admits a solution $x \in \RR^{V \times \II1{d-1}}$.	
	According to Farkas lemma, it thus remains to show that ${}^t A_R y = 0$ admits a non-zero solution $y \in (\RR^+)^{E \times \II1d}$ if and only if  $R$ admits a non-zero multi-flow.
	
	Suppose that the system ${}^t A_R y = 0$ admits a non-zero solution $y \in (\RR^+)^{E \times \II1d}$.
	Then $y_{(e,i)} \ge 0$ for all $i \in \II1d$ and all $e \in E$.
	For any $i \in \II1d$, we define the flow $\varphi_i$ on $G^i(R)$ by setting $\varphi_i(a) = y_{(e,i)}$ for every $a \in A^i$, where $e$ is the edge of $E$ corresponding to the arc $a$.
	Since ${}^t A_R y = 0$, the following equations hold for every $v\in V$ and every $j \in \II1{d-1}$, 
	\begin{eqnarray*}
	    \sum_{ (e,i) \in E \times \II1d} a_{(e,i),(v,j)} y_{(e,i)} & = & 0\\
	    \sum_{ (e,i) \in E \times \II1{d-1}} a_{(e,i),(v,j)} y_{(e,i)} & = & - \sum_{ (e,i) \in E \times \{d\}} a_{(e,i),(v,j)} y_{(e,i)}\\
	\end{eqnarray*}
	Since $a_{(e,i),(v,j)} = 0$ whenever $i\neq j$ and $i\neq d$,
	\begin{eqnarray*}
        \sum_{e \in E} a_{(e,j),(v,j)} y_{(e,j)} & = & - \sum_{e \in E} a_{(e,d),(v,j)} y_{(e,d)}\\
	\end{eqnarray*}
    By definition of $A_R$,
	\begin{eqnarray*}
        \sum_{\substack{e=\{u,v\} \in E\\ \text{s.t.}\; u\le_j v}} y_{(e,j)} - \sum_{\substack{e=\{u,v\} \in E\\ \text{s.t.}\; v\le_j u}} y_{(e,j)} & = & \sum_{\substack{e=\{u,v\} \in E\\ \text{s.t.}\; u\le_d v}} y_{(e,d)} - \sum_{\substack{e=\{u,v\} \in E\\ \text{s.t.}\; v\le_d u}} y_{(e,d)}\\
	\end{eqnarray*}
    Finally by definition of $\varphi_i$, 
	\begin{eqnarray*}
        \sum_{(u,v) \in A^j} \varphi_j(u,v) - \sum_{(v,u) \in A^j} \varphi_j(v,u) & = & \sum_{(u,v) \in A^d} \varphi_d(u,v) - \sum_{(v,u) \in A^d} \varphi_d(v,u)\\
	    \DD_{\varphi_j}(v) & = & \DD_{\varphi_d}(v)
	\end{eqnarray*}
	We conclude that $(\varphi_1, \ldots , \varphi_d)$ is a non-zero multi-flow of $R$.
	
	To prove the converse statement suppose that $R$ admits a non-zero multi-flow $(\varphi_1, \ldots , \varphi_d)$.
	We define $y \in \RR^{E \times \II1d}$ by setting $y_{(e,i)}$ by $y_{(e,i)} = \varphi_i(a)$, for any $e \in E$ and $i \in \II1d$, and where $a$ is the arc of $G^i(R)$ corresponding to the edge $e$.
	Clearly $y \in (\RR^+)^{E \times \II1d}$, and the multi-flow being non-zero, $y$ is non-zero. As
	$\DD_{\varphi_j}(v) = \DD_{\varphi_d}(v)$ for every $v\in V$ and every $j\in\II1{d-1}$, one can deduce (by reversing the above calculus) that ${}^t A_R y = 0$. We thus have that $y$ is a non-zero solution to ${}^t A_R y = 0$.
\end{proof}

\section{A counter-example to Conjecture~\ref{conj}}

Before showing our counter-example to Conjecture~\ref{conj} we need the following definition.
\begin{definition}
A $d$-representation $R$ on a set $V$ with at least $d$ elements, is {\it standard} if every element $v\in V$ is a vertex of $\Sigma(R)$, and if for any order, its maximal vertex is among the $d-1$ smallest elements in all the other orders of $R$.
\end{definition}
Let us recall that in a simplicial complex, a {\it $k$-face} $F$ is a face such that $|F|=k+1$.
\begin{lemma}\label{lem:standard}
Let $R$ be a $d$-representation on $V$ such that every element of $V$ is a vertex of $\Sigma(R)$.
The representation $R$ is standard if and only if there exists vertices $M_1, \ldots M_d$ such that in $\Sigma(R)$ every face belongs to at least one $(d-1)$-face, every $(d-2)$-face belongs to at least two $(d-1)$-faces except the $(d-2)$-faces $F_{i} = \{M_1,\ldots M_d\}\setminus M_i$ which belong to only one $(d-1)$-face.
\end{lemma}
\begin{proof}
($\Longrightarrow$) Clear from~\cite{Ossona}.

($\Longleftarrow$) Suppose that there exists vertices $\{M_1, \ldots, M_d\}$ such that every face belongs to at least one $(d-1)$-face, every $(d-2)$-face belongs to at least two $(d-1)$-faces except the $(d-2)$-faces $F_{i} = \{M_1,\ldots M_d\}\setminus M_i$ which belong to only one $(d-1)$-face.


For any $i \in \II1d$ we define $F'_i$ as a $(d-1)$-face minimizing $\max_{\le_i}\{f\in F'_i\}$ in $\le_i$ among the other $(d-1)$-faces. Let us denote $f'_{i,j}$ the element of $F'_i$ that is maximal in $\le_j$. As each of the $d$ elements of $F'_i$ dominates it at least once, $F'_i = \{f'_{i,1},\ldots,f'_{i,d}\}$. Let us show that $F'_i\setminus\{f'_{i,i}\}$, which is a $(d-2)$-face, does not belong to any other $(d-1)$-face. Indeed, if there exists a vertex $x\neq f'_{i,i}$ such that $F'_i\setminus\{f'_{i,i}\}\cup\{x\}$ is a $(d-1)$-face then $F'_i<_i x$ by definition of $F'_i$, and $f'_{i,i}$ cannot dominate $F'_i\setminus\{f'_{i,i}\}\cup\{x\}$ (as $f'_{i,i} <_i x$ and $f'_{i,i}<_j f'_{i,j}$ for any $j\neq i$), a contradiction.

Suppose that there exists $i$ and $j \neq i$ such that $F'_i \setminus \{f'_{i,i} \} = F'_j \setminus \{f'_{j,j}\}$.
We define $X = F'_i \setminus \{f'_{i,i}\}$.
As $X$ is included in $F'_i$ and in $F'_j$ , and as it belongs to only one $(d-1)$-face $F'_i$, we have that $F'_i = F'_j$ and thus that $f'_{i,i} = f'_{j,j}$.
Then in the $(d-1)$-face $F'_i = F'_j$, the vertex $f'_{i,i} = f'_{j,j}$ thus dominates $F'_i$ in two orders, $<_i$ and $<_j$, and one of the remaining $d-1$ elements of $F'_i$ cannot dominate it, a contradiction.
We conclude that the faces $F'_i \setminus \{f'_{i,i}\}$ are distinct and are in bijection with the faces $F_i$.
We can thus assume that $F_{i} = F'_i \setminus \{f'_{i,i}\}$. 

We now show that the $M_i$'s are the maxima of the representation.
Let $i \in \II1d$ and $j \not= i$.
Since $M_i \in F_j$, none of the $(d-1)$-faces is dominated by $M_i$ in the order $<_j$. Indeed if it was the case then we would have a $(d-1)$-face $F''$ such that $F'' \le_j M_i <_j f'_{j,j}$ which contradicts the definition of $F'_j$.
Thus $F'' \le_i M_i$ for every $(d-1)$-face $F''$.
For every element $x$, as $x$ is included in at least one $(d-1)$-face, we have that $x \le_i M_i$.
We conclude that $M_i$ is the maximum of $\le_i$.

We now show that $M_i$ is among the $d-1$ smallest elements of the order $\le_j$ for every $j \not= i$.
Since the face $F_j$ contains $M_i$ for every $i \not= j$, none of the elements in $V\setminus F_j$ dominates $F_j$ in $\le_i$.
Thus each of the $|V|-(d-1)$ elements of $V\setminus F_j$ dominates $F_j$ in the order $\le_j$. $M_i$ is thus among the $d-1$ smallest elements in order $\le_j$.
The representation $R$ is thus standard.
\end{proof}

\begin{theorem}\label{thm:main}
	Let $R$ be the following 4-representation on $V=\{a,b,c,d,e,f,g,h\}$:

	$$
	\begin{array}{c|ccccccccccccc}
	\le_1 & b & c & d & e & g & f & h & a \\
	\le_2 & a & c & d & e & h & f & g & b \\
	\le_3 & a & b & d & f & g & e & h & c \\
	\le_4 & a & b & c & f & h & e & g & d \\
	\end{array}$$

	The simplicial complex $\Delta = \Sigma(R)$ has Dushnik-Miller dimension $4$ but it is not a TD-Delaunay complex of $H_4\simeq \RR^{3}$.
\end{theorem}

\begin{proof}
Let us first show that any $4$-representation $R'$ on $V$
such that $\Sigma(R') = \Delta$ is equivalent to $R$ up to permutations of the orders and up to a permutation of the smallest $3$ elements in each order. 

By Lemma~\ref{lem:standard} such $R'$ is standard with maximal elements $a, b, c$, and $d$. Without loss of generality we assume that these elements are maximal in $\le'_1, \le'_2, \le'_3$, and $\le'_4$, respectively. 
As $\{e,b,c,d\}\in \Delta$ and as none of $a, e, f, g, h$ dominates this face in $\le'_2, \le'_3$, or $\le'_4$, $e$ is necessarily the fourth smallest element in $\le'_1$. We similarly deduce that $e$ is also the fourth smallest element in $\le'_2$ (using the face $\{a,c,d,e\}$ of $\Delta$), and that $f$ is the fourth smallest element in $\le'_3$ and $\le'_4$ (using the faces $\{a,b,d,f\}$ and $\{a,b,c,f\}$ respectively). 
As $f$ dominates $\{b,e,g\}\in \Delta$, we have that $g\le'_1 f$. We similarly deduce that $h\le'_2 f$, $g\le'_3 e$, and $h\le'_4 e$ (using faces $hea$, $gdf$, and $chf$).
As $h$ dominates $\{c,e,f\}\in \Delta$, we have that $f\le'_1 h$. We similarly deduce that $f\le'_2 g$, $e\le'_3 h$, and $e\le'_4 g$ (using faces $def$, $aef$, and $bef$).

This implies that for any such $R'$ the subdigraphs of $G^i(R')$ induced by the vertices $e, f, g,$ and $h$ are the same as in $G^i(R)$.
These digraphs are depicted in Figure~\ref{fig:cex_multi_flow}, as well as 
a multi-flow of $R'$. Vertices $a, b, c,$ and $d$ are not drawn as the flows on their incident edges are null. Their divergences are thus $0$.
The divergences of the vertices $e$ and $f$ is $-1$ and it is $+1$ for $g$ and $f$.
	
	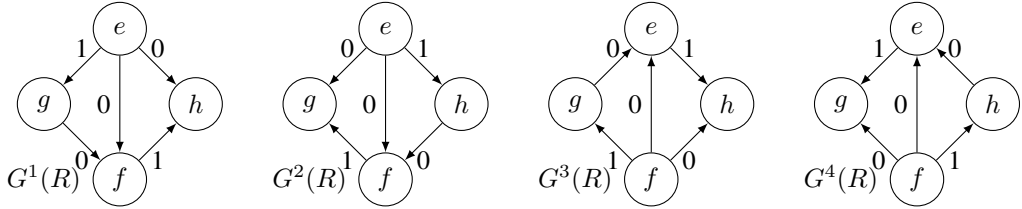
\begin{figure}[h]
		\label{fig:cex_multi_flow}
		\centering
		\begin{tikzpicture}
		\def\x{1	}
		\def\s{3.5}
		
		\begin{tabular}{c}
		 
		\draw node at (-1,-1) {$G^1(R)$};
		\draw node at (2.5,-1) {$G^2(R)$};
		\draw node at (6,-1) {$G^3(R)$};
		\draw node at (9.5,-1) {$G^4(R)$};

		\foreach \i in {0,1,2,3}
		{
		\draw node[draw, circle, minimum size=0.7cm] (a\i) at (\i*\s,\x)  {$e$};
		\draw node[draw, circle, minimum size=0.7cm] (b\i) at (\i*\s,-\x)  {$f$};
		\draw node[draw, circle, minimum size=0.7cm] (c\i) at (\i*\s-\x,0)  {$g$};
		\draw node[draw, circle, minimum size=0.7cm] (d\i) at (\i*\s+\x,0)  {$h$};
		}

		\draw[->,>=latex] (a0) -- (c0) node[midway,above] {1};
		\draw[->,>=latex] (b0) -- (d0) node[midway,below] {1};
		\draw[->,>=latex] (a0) -- (b0) node[midway,left] {0};
		\draw[->,>=latex] (a0) -- (d0) node[midway,above] {0};
		\draw[->,>=latex] (c0) -- (b0) node[midway,below] {0};

		\draw[->,>=latex] (a1) -- (d1) node[midway,above] {1};
		\draw[->,>=latex] (b1) -- (c1) node[midway,below] {1};
		\draw[->,>=latex] (a1) -- (b1) node[midway,left] {0};
		\draw[->,>=latex] (a1) -- (c1) node[midway,above] {0};
		\draw[->,>=latex] (d1) -- (b1) node[midway,below] {0};

		\draw[->,>=latex] (a2) -- (d2) node[midway,above] {1};
		\draw[->,>=latex] (b2) -- (c2) node[midway,below] {1};
		\draw[->,>=latex] (b2) -- (a2) node[midway,left] {0};
		\draw[->,>=latex] (b2) -- (d2) node[midway,below] {0};
		\draw[->,>=latex] (c2) -- (a2) node[midway,above] {0};

		\draw[->,>=latex] (a3) -- (c3) node[midway,above] {1};
		\draw[->,>=latex] (b3) -- (d3) node[midway,below] {1};
		\draw[->,>=latex] (b3) -- (a3) node[midway,left] {0};
		\draw[->,>=latex] (d3) -- (a3) node[midway,above] {0};
		\draw[->,>=latex] (b3) -- (c3) node[midway,below] {0};
		\end{tabular}
		\end{tikzpicture}
		\caption{Definition of the multi-flow on $R'$.}
	\end{figure}
Thus, as every $4$-representation $R'$ such that $\Delta = \Sigma(R')$ has a multi-flow, Proposition~\ref{proposition:alternative} and Theorem~\ref{thm:delaunay_ineq_syst} imply that $\Delta$ is not a TD-Delaunay complex of $H_4\simeq \RR^3$.
\end{proof}

\section{Link with Rectangular Delaunay complexes}

Rectangular Delaunay graphs have been studied independently by Felsner~\cite{FarXiv06} and by Chen \textit{et al.}~\cite{CPST}.
In the following, $\PP$ denotes a finite set of points of $\RR^2$ such that no two points share the same vertical or horizontal coordinate.

\begin{definition}
    We define the {\it R-Delaunay complex} $\RD(\PP)$ of the point set $\PP$ as the simplicial complex whose vertex set is $\PP$ such that a subset $F$ of $\PP$ forms a face if there exists an axis-parallel rectangle $R$ such that $R \cap \PP = F$ and such that $R$ does not contain any point of $\PP$ in its interior. The {\it R-Delaunay graph} of a point set $\PP$ is the graph defined by the faces of size one and two in $\RD(\PP)$.
\end{definition}

\begin{figure}[!h]
    \centering
    \includegraphics[scale=1]{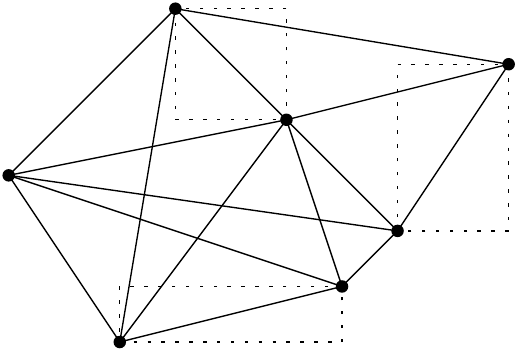}
    \caption{An example of a R-Delaunay graph where three rectangles are drawn in dashed segments.}
    \label{fig:my_label}
\end{figure}

R-Delaunay graphs have interesting properties. Felsner~\cite{FarXiv06} showed that those graphs can have a quadratic number of edges. In order to solve a question motivated by a frequency assignment problem in cellular telephone networks and related to conflit-free colorings~\cite{ELRS}, Chen {\it et al.} showed the following result.

\begin{theorem}[Chen {\it et al.}~\cite{CPST}]
    Given points $\PP$ in the unit square selected randomly and uniformly,
    the probability that the largest independent set of $RD(\PP)$ is $O(n \log^2 \log(n)/\log(n)$) tends to $1$.
\end{theorem}
This implies that R-Delaunay graphs
can have arbitrarily large chromatic number.

Given a point set $\PP$, we define the horizontal order $\le_1$ on $\PP$ as follows: for every two points $x=(x_1,x_2)$ and $y=(y_1,y_2) \in \PP$, $x \le_1 y$ if and only if $x_1 \le y_1$.
Then we define $\le_2$ as the reverse of the horizontal order.
We define the vertical order $\le_3$ on $\PP$ as previously: here $x \le_3 y$ if and only if $x_2 \le y_2$; and we define $\le_4$ as the reverse of the vertical order.

One can show that using these four orders one obtains a representation $R$ such that $\RD(\PP) = \Sigma(R)$ and actually any 4-representation,  where two pairs of orders are the reverse from each other, defines   an R-Delaunay complex. Felsner~\cite{FarXiv06} showed this for graphs but his proof easily extends to simplicial complexes. In his words, R-Delaunay complexes are exactly the complexes of dimension $[3\updownarrow\updownarrow4]$. This implies that R-Delaunay complexes form a subclass of the simplicial complexes with Dushnik-Miller dimension at most 4. The following theorem refines this by showing that it is also a subclass of TD-Delaunay complexes of $H_4\simeq \RR^3$.
    
\begin{theorem}\label{thm:dim34}
    The class of R-Delaunay complexes is included in the class $\TDD_4$ of TD-Delaunay complexes of $H_4\simeq \RR^3$.
\end{theorem}
\begin{proof}
    We want to show that $RD(P)$ is a TD-Delaunay complex.
    According to Theorem~\ref{thm:delaunay_ineq_syst}, it is enough to show that the TD-Delaunay system of inequalities defined by $R$, where $R$ is the previously defined 4-representation, has a solution. By Proposition~\ref{proposition:alternative}, this is equivalent to show that there is no non-zero multi-flow  $\varphi = (\varphi_1, \varphi_2, \varphi_3, \varphi_4)$ on $R$.

    Consider $x$ the rightmost point: it is the greatest point according to order $\le_1$.
    So $x$ is a sink in $G^1(R)$ and $\DD_{\varphi_1}(x) \ge 0$ for any flow $\varphi_1$ on $G^1(R)$.
    As $\le_2$ is order $\le_1$ reversed, $x$ is also the minimum in order $\le_2$.
    So $x$ is a source in $G^2(R)$ and $\DD_{\varphi_2}(x) \le 0$ for any flow $\varphi_2$ on $G^2(R)$.
    So if $\DD_{\varphi_1}(x) = \DD_{\varphi_2}(x)$ we have that this divergence is null, as well as $\varphi_1$ and $\varphi_2$ on the arcs incident to $x$.
    By induction on the number of vertices, we show that all the divergences defined by $\varphi_1$ (and $\varphi_2$) are null. Note that a flow $\varphi$ defined on an acyclic digraph (such as $G^1(R), G^2(R), G^3(R)$, and $G^4(R)$) with divergence zero at every vertex, is null everywhere (i.e. for each arc $a$, $\varphi(a)=0$). Therefore $R$ has no non-zero multi-flow, and the system of inequalities defined by $R$ admits a solution and $RD(\PP)=\Sigma(R)$. The simplicial complex $RD(\PP)$ is thus a TD-Delaunay complex of $H_4\simeq \RR^3$.
\end{proof}

Note that these classes are distinct as $K_5\in \TDD_4$ is not an R-Delaunay graph.
Another way to prove the above theorem is by showing that in $\mathbb{R}^3$, any rectangle with sides parallel to the axis and drawn on the plane defined by $z=0$, is the intersection between this plane and a regular simplex homothetic to the simplex with vertices $(-2,0,+\sqrt{2}), (0,-2,-\sqrt{2}), (+2,0,+\sqrt{2}), (0,+2,-\sqrt{2})$. Such a proof would also be simple but it would involve a few calculus that were avoided in the proof above.

Note that the arguments in the proof of Theorem~\ref{thm:dim34} show that any $d$-representation $R$, where two orders are the reverse from each other, is such that its TD-Delaunay system has a non-zero solution and thus $\Sigma(R)$ is a TD-Delaunay complex of $H_d\simeq \RR^{d-1}$.



\section{Conclusion}
\label{sec:ccl}

TD-Delaunay complexes of $\RR^3$ have shown their interest as sparse spanners with small stretch~\cite{BGHI}. 
It is likely that this extends to higher dimensional spaces. 
A question is for example whether graphs of Dushnik-Miller dimension 4 give better spanners than TD-Delaunay complexes in $\RR^3$.

Furthermore, as the class of TD-Delaunay complexes of $H_d\simeq \RR^{d-1}$ is strictly included in the class of complexes of Dushnik-Miller dimension $d$, it may be interesting to study the problem of the grid embedding for this restricted class.  
Indeed Schnyder showed~\cite{Schnyder} that a planar graph with $n$ vertices can be embedded in an $n\times n$ grid (without crossings). Ossona de Mendez showed~\cite{Ossona} that a complex of Dushnik-Miller dimension $d$ could be drawn in $\RR^{d-1}$ without crossings. 
Nevertheless the last embedding uses exponentially large grids and it is still an open problem to reduce the sizes of these grids.
Is it possible to reduce the size of these grids when dealing with TD-Delaunay complexes?



\section*{Acknowledgments}
The authors are thankful to Arnaud Mary and Vincent Pilaud  for interesting discussions on this topic and to Gwenaël Joret for pointing us the class of rectangular Delaunay complexes.

\end{document}